\documentclass[9pt,a4paper]{article}

\usepackage{lmodern}
\usepackage[latin1]{inputenc}
\usepackage{graphicx,color}
\usepackage{epsfig}
\usepackage{verbatim}
\usepackage[english]{babel}
\usepackage{amsmath, amsthm, amssymb}
\usepackage{wrapfig}
\usepackage{neubauer}

\usepackage{enumerate}
\usepackage{csquotes}
\usepackage{placeins}
\usepackage{appendix}

\newtheorem{prop}{Proposition}
\newtheorem{theorem}{Theorem}

\makeatletter
\renewcommand{\@fnsymbol}[1]{\@arabic{#1}}
\makeatother

\author{Victoria Hutterer\footnote{Industrial Mathematics Institute, Johannes Kepler University Linz, Altenbergerstrasse 69, 4040 Linz, Austria. victoria.hutterer@indmath.uni-linz.ac.at} \ and Ronny Ramlau\footnotemark[1]\ $^{,}$\footnote{Johann Radon Institute for Computational and Applied Mathematics (RICAM) Linz, Altenbergerstrasse 69, 4040 Linz,  Austria.}}
\title{Non-linear wavefront reconstruction methods for pyramid sensors using Landweber and Landweber-Kaczmarz iteration}

\begin{document}

\maketitle

\begin{abstract}
Accurate and robust wavefront reconstruction methods for pyramid wavefront sensors are in high demand as these sensors are planned to be part of many instruments currently under development for ground based telescopes. The pyramid sensor relates the incoming wavefront and its measurements in a non-linear way. Nevertheless, almost all existing reconstruction algorithms are based on a linearization of the model. The assumption of a linear pyramid sensor response is justifiable in closed loop AO when the measured phase information is small but may not be reasonable in reality due to unpreventable errors depending on the system such as non common path aberrations. In order to solve the non-linear inverse problem of wavefront reconstruction from pyramid sensor data we introduce two new methods based on the non-linear Landweber and Landweber-Kaczmarz iteration. Using these algorithms we experience high-quality wavefront estimation especially for the non-modulated sensor by still keeping the numerical effort feasible for large-scale AO systems.
\end{abstract}


\section{Introduction}
Time-varying optical perturbations introduced by the atmosphere severely degrade the image quality of ground based telescopes. Adaptive Optics (AO) systems correct these atmospheric aberrations in real-time \cite{Ha98,Roddier,Tyson00}: The facilities have devices incorporated that sense the incoming wavefronts and cancel the originated perturbations with a deformable mirror. Suitable mirror configurations are based on an accurate estimation of the shape of the incoming wavefront and can be calculated from wavefront sensor measurements. Reconstruction of the wavefront from sensor data is an inverse problem for which the underlying mathematical forward model depends on the type of the wavefront sensor (WFS). This paper is concerned with solving the inverse problem of wavefront reconstruction using a pyramid wavefront sensor. \bigskip

More than twenty years ago, the pyramid wavefront sensor (PWFS) was proposed for the first time as a promising alternative to other types of wavefront measuring devices \cite{Raga96}. For the next generation of ground based Extremely Large Telescopes (ELTs), the pyramid sensor has been gaining attention from the scientific community by setting new standards for AO correction quality. Due to outstanding results recorded on existing observing facilities, the sensor is included as baseline for many instruments on ELTs such as ESO's Extremely Large Telescope, the Thirty Meter Telescope (TMT), and the Giant Magellan Telescope (GMT). Particularly for segmented ELT pupils, the pyramid sensor seems to be the wavefront sensor of choice \cite{Esposito_2003_SPIE_pwfs_cophasing,
Esposito_2012_pwfs_NGS_SCAO_GMT,Hippler_2018,Surdej_Thesis}. \bigskip

A focus lies in the development of wavefront reconstruction algorithms for the pyramid sensor based on theoretical investigations of the model as, e.g., in \cite{HuSha18_1,KoVe07,Shatokhina_PhDThesis,Veri04} accompanied by numerical simulations or laboratory studies on optical test benches \cite{Bond_2016,EsRiFe00,Martin_ao4elt4_pwfs_bench_on-sky,Pinna_2008_HOT_bench,RaFa99,Turbide_ao4elt3_pwfs_bench,Veran_ao4elt4_pwfs_vs_sh,
Verinaud_2005_pwfs_vs_sh,Ragazzoni_ao4elt3_pwfs_lab}.
On sky, the first single star AO loop was closed on AdOpt@TNG \cite{raga2000a} at the Telescopio Nazionale Galileo with a PWFS. In recent years, remarkable operational results have been reported at the $8$~m Large Binocular Telescope (LBT) \cite{Esposito2010_AO_for_LBT,Esposito_2011_pwfs_onsky,Pedichini_2016_XAO_pwfs_LBT}. In addition to the LBT, the pyramid sensor is integrated in the AO systems of the Subaru Telescope (SCexAO), the Magellan Telescope (MagAO), the Mont Megantic Telescope (INO Demonstrator), and the Calar Alto Telescope (PYRAMIR).  \bigskip

Aside from astronomical applications, the pyramid sensor is used in adaptive loops in ophthalmology \cite{Chamot06,Daly_2010,Iglesias02} and microscopy \cite{Ig11, Ig13}. The underlying concepts are comparable to atmosphere induced perturbations sensing for adaptive optics in astronomy. In microscopy, the pyramid sensor is introduced for direct phase detection. Unstained cellular media sometimes appear transparent. Hence, measuring the imprinted phase changes induced by variations in the index of refraction is necessary for the observation of biological structures. For adaptive optics systems in the eye, the pyramid sensor is used to perform high efficient and flexible wavefront sensing to compensate ocular aberrations.  \bigskip 

From a mathematical perspective, the relation between the incoming, unknown wavefront and the measured pyramid wavefront sensor response is non-linear. 
Basically, the pyramid sensor signal can be modeled as the incoming wave convolved with the point spread function of the sensor. Due to the sinusoidal nature of the measurements, the model has a suitable linear approximation depending on the amplitude of the incoming wavefront \cite{HuSha18_1}. In closed loop AO, already corrected wavefronts are measured by the wavefront sensor. Thus, the existence of small incident wavefronts allows to assume a linear response of the PWFS. However, for open loop data or larger wavefront errors, for instance induced by non common path aberrations (NCPAs), this assumption is not fulfilled. Non common path errors appear in AO systems when the wavefront sensor does not belong to the same light path as the science camera. Then, the AO system suffers from aberration differences between the wavefront sensor and the science camera. Additionally, optical elements may be incorporated in the non common path as, for instance, in the case of Multi Conjugate AO systems. In these cases the assumption of small residual wavefronts being measured by the wavefront sensor, and further the linearity of the pyramid sensor may be violated. Non-linear wavefront reconstructors are considered as one possibility to handle the non-linearity effects of the pyramid sensor introduced by influences such as NCPAs. Currently, there still exist only a few attempts, e.g., \cite{Clare_2004,Frazin_18,Fauv16,Fauv17,korkiakoski_nonlinear_08,Korkiakoski_08,KoVe07,Viotto16} for handling the non-linearity of pyramid sensors with applications in astronomical AO. 

We introduce a new idea of non-linear wavefront reconstruction and propose to apply the non-linear Landweber/Landweber-Kaczmarz method. Both iterative algorithms have already been studied in-depth by the mathematical community with multiple applications in the field of inverse problems. \bigskip

The paper is organized as follows: 

The non-linear inverse problem of wavefront reconstruction from pyramid sensor data is introduced in Section \ref{chap_invProb}. We also consider the mathematical models of the underlying wavefront sensor as well as its approximations. In Section \ref{chap_landweber}, we adapt the Landweber and Landweber-Kaczmarz iteration to the problem of non-linear wavefront reconstruction using pyramid sensors and introduce two new methods for wavefront reconstruction, namely the non-linear LIPS (Landweber Iteration for Pyramid Sensors) and the non-linear KLIPS (Kaczmarz Landweber Iteration for Pyramid Sensors). The evaluation of the Fr\'{e}chet derivatives and the corresponding adjoints which are needed for the application of the algorithms is done in Section \ref{chap_frechet}. In Section \ref{chap_complexity}, we mention some details on the discretization of the problem and the computational complexity of the proposed methods, and we show the performance of the reconstructors using closed loop end-to-end simulations in Section \ref{chap_numerics}.

\section{Non-linear problem of wavefront reconstruction using pyramid sensors} \label{chap_invProb}

For wavefront reconstruction from pyramid sensor data the aim is to solve the non-linear pyramid sensor operator equation
\begin{equation}\label{eq:n.1}
s= \boldsymbol P \Phi,
\end{equation}
where $s=\left[s_x,s_y\right]$ denotes pyramid wavefront sensor measurements, $\boldsymbol P:\mathcal{D}\left(\boldsymbol P\right)\rightarrow \mathcal{L}_2\left(\mathbb{R}^2\right)$ the non-linear pyramid sensor operator with $\mathcal{D}\left(\boldsymbol P\right) \subseteq \mathcal{H}^{11/6}\left( \mathbb{R}^2\right)$ and $\Phi \in \mathcal{H}^{11/6}\left(\mathbb{R}^2\right)$ the unknown incoming wavefront \cite{Ellerbroek02,Eslitz13,HuSha18_1,Neub12b}. Due to the finite size of both the telescope pupil and the wavefront sensor detector, in the following denoted by $\Omega=\Omega_y\times\Omega_x$ for simplicity, the involved wavefronts $\Phi$ and sensor measurements $s$ have compact support on $\Omega$. The indices $x$ and $y$ in $\Omega_y\times\Omega_x$ indicate the dependence  on $x$ and $y$ of single lines (intervals) in the telescope aperture having annular shape. Depending on $x$ or $y$, the chords either consist of one or two parts and have varying lengths. The norms in the Hilbert spaces are considered to be the norms in $\mathcal{L}_2\left(\mathbb{R}^2\right)$ with $$\left|\left|\cdot\right|\right|_{\mathcal{L}_2\left(\mathbb{R}^2\right)}=\left|\left|\cdot\right|\right|_{\mathcal{L}_2\left(\Omega\right)}$$
because of the compact support of the involved functions  on the telescope aperture and are generally denoted by $\left|\left|\cdot\right|\right|$ throughout the paper.

Since the unperturbed data $s$ are almost never available in practice, we consider noisy data $s^\delta$ and assume $$\left|\left|s^\delta-s\right|\right|<\delta$$ for some noise level $\delta > 0$. \bigskip

We start by identifying pyramid sensor operators and derive their underlying theoretical principles. The physical configuration of the sensor is based on a Foucault knife edge and the core characteristic of a pyramid sensor is a refracting pyramidal prism of four facets that split the beam. In our wavefront reconstruction approach, we will additionally consider (as a simplification of the pyramid sensor) the roof wavefront sensor consisting of two orthogonally placed roof prisms instead of one pyramidal prism. The mathematical roof WFS model constitutes a part of the pyramid model as shown in the next Section. 

In this paper we focus mainly on a mathematical description of both sensors. Hence, we will not explain the physical principles of pyramid or roof sensors in detail and refer the reader to \cite{Raga96,Ricc98} for more precise explanations. Measurements of both wavefront sensors always consist of two data sets $s=\left[s_x,s_y\right]$, one in $x$-direction and one in $y$-direction. Due to the structure of the measurement equations introduced later in Theorem \ref{pyr_def}, we define the pyramid sensor operators $ \boldsymbol P^{\{n,c\}} := [-\tfrac{1}{2}\boldsymbol P_x^{\{n,c\}},\tfrac{1}{2}\boldsymbol P_y^{\{n,c\}}]$  where the index indicates $x$- or $y$-direction and the superscript $n$ the non-modulated sensor or $c$ the circular modulated sensor. The corresponding roof sensor operators are denoted by $\boldsymbol R^{\{n,c\}} := [-\tfrac{1}{2}\boldsymbol R_x^{\{n,c\}},\tfrac{1}{2} \boldsymbol R_y^{\{n,c\}}]$. Whenever we omit the superscripts $\{n,c\}$, then the theory is applicable to both non-modulated and circular modulated pyramid sensor as well as the linear modulated roof wavefront sensor which is not specifically reviewed in this paper. The domain $\mathcal{D}\left(\boldsymbol P\right)$ of the pyramid operator either denotes $\mathcal{D}\left(\boldsymbol P_x\right)$ or $\mathcal{D}\left(\boldsymbol P_y\right)$ respectively depending on the direction we are considering the problem (for Landweber iteration) or $\mathcal{D}\left(\boldsymbol P_x\right) \cap \mathcal{D}\left(\boldsymbol P_y\right)$ if we consider the problem as a system of equations (in case of Landweber-Kaczmarz iteration).

%
\subsection{Pyramid forward models without interference effects}
%
%

In what follows we use the analytical pyramid wavefront sensor transmission mask model \cite{Feeney_PhDThesis,KoVe07,Shatokhina_PhDThesis,Veri04}. In case of circular modulation, we denote the modulation parameter by
\begin{equation}\label{eq:mod_param}
 \alpha_{\lambda} = \left(2\pi \alpha\right)/ \lambda ,
\end{equation}
with $\alpha = r\lambda/D$ for a positive integer $r$ representing the modulation radius, $\lambda$ the sensing wavelength, and $D$ the telescope diameter.

\begin{theorem}\label{pyr_def}
The pyramid wavefront sensor data in the non-linear transmission mask model are represented as
\begin{align*}
s_x^{\{n,c\}}(x,y) &= -\tfrac{1}{2}\left(\boldsymbol P_x^{\{n,c\}}\Phi\right)(x,y), \\ \notag
s_y^{\{n,c\}}(x,y) & = \tfrac{1}{2}\left(\boldsymbol P_y^{\{n,c\}}\Phi\right)(x,y),
\end{align*}
where the operators $\boldsymbol P_x^{\{n,c\}}:\mathcal{H}^{11/6}\left(\mathbb{R}^2\right)\rightarrow
\mathcal{L}_2\left(\mathbb{R}^2\right)$ in $x$-direction are given by
{\small\begin{align*}
\left(\boldsymbol P_x^{\{n,c\}}\Phi\right)(x,y)&:=\mathcal{X}_{\Omega}(x,y)\left.\dfrac{1}{\pi} \int_{\Omega_y}{\dfrac{\sin{\left[\Phi(x',y)-\Phi(x,y)\right] \times k^{\{n,c\}} (x'-x)  }}{x'-x}\ \mathrm{d}x'}\right. \\ \notag
&\left.+\mathcal{X}_{\Omega_y}(x) \dfrac{1}{\pi^3} \ p.v. \int_{\Omega_y}{\int_{\Omega_x}{\int_{\Omega_x}{\dfrac{\sin{\left[\Phi(x',y')-\Phi(x,y'')\right] \times l^{\{n,c\}} (x'-x,y''-y') }}{(x'-x)(y'-y)(y''-y)} \ \mathrm{d}y'' \ }\mathrm{d}y' \  } \mathrm{d}x' }\right.
\end{align*}
}and the operators $\boldsymbol P_y^{\{n,c\}}: \mathcal{H}^{11/6}\left(\mathbb{R}^2\right)\rightarrow
\mathcal{L}_2\left(\mathbb{R}^2\right)$ in $y$-direction are given by
{\small\begin{align}\label{eq:3.3} \notag
\left(\boldsymbol P_y^{\{n,c\}}\Phi\right)(x,y)&:=\mathcal{X}_{\Omega}(x,y)\left.\dfrac{1}{\pi} \int_{\Omega_x}{\dfrac{\sin{\left[\Phi(x,y')-\Phi(x,y)\right] \times k^{\{n,c\}} (y'-y) }}{y'-y}\ \mathrm{d}y'} \right.\\ \notag
&\left.+\mathcal{X}_{\Omega_x}(y) \dfrac{1}{\pi^3}\ p.v. \int_{\Omega_y}{\int_{\Omega_x}{\int_{\Omega_y}{\dfrac{\sin{\left[\Phi(x',y')-\Phi(x'',y)\right. \times l^{\{n,c\}} (x''-x',y'-y) }}{(x'-x)(y'-y)(x''-x)} \ \mathrm{d}x'' \ }\mathrm{d}y' \ } \mathrm{d}x' }\right].
\end{align}
}The functions $k^{\{n,c\}}$ are defined by $k^n (x) := 1 , \ k^c (x): = J_0 (\alpha_{\lambda} x) $,
and the functions $l^{\{n,c\}}$ by $l^n(x,y): = 1$ and
\begin{equation*}
 l^c(x,y) := \dfrac{1}{T} \int_{-T/2}^{T/2} \cos [ \alpha_{\lambda} x \sin (2\pi t/T) ] \cos [ \alpha_{\lambda} y \cos (2\pi t/T) ] \ \mathrm{d}t .
\end{equation*}
The function $J_0$ denotes the zero-order Bessel function of the first kind, i.e.,
\begin{equation*}
 J_0 (x) = \dfrac{1}{\pi} \int_0^{\pi} \cos ( x \sin t ) \ \mathrm{d}t 
\end{equation*}
with modulation parameter $\alpha_{\lambda}$ defined in~\eqref{eq:mod_param}.
\end{theorem}

\begin{proof}
An optics-based derivation is given in \cite{KoVe07}, or a more mathematical version using distribution theory in \cite{HuSha18_1}.
\end{proof}

Due to the involved singularities, the integrals are defined in the Cauchy principal value ($p.v.$) meaning. The principal value stems from the pyramid transmittance function incorporated as a description of the pyramidal prism in the derivation of the forward model \cite{HuSha18_1,KoVe07,Veri04}.

 
The roof sensor operators consist of the first term of the pyramid sensor operators. For the development of model-based reconstruction algorithms, the roof sensor is either considered as a standalone wavefront sensor or used as a simplification of the full pyramid sensor model.

\begin{theorem}
The forward models of the non-modulated and modulated roof WFS in the transmission mask approach are represented by
\begin{align*}
s_x^{\{n,c\}} (x,y) &= -\tfrac{1}{2}\left( \boldsymbol R_x^{\{n,c\}} \Phi \right) (x,y), \\
s_y^{\{n,c\}} (x,y) &= \tfrac{1}{2}\left( \boldsymbol R_y^{\{n,c\}} \Phi \right) (x,y). 
\end{align*}
with the operators $\boldsymbol R_x^{\{n,c\}}:\mathcal{H}^{11/6}\left(\mathbb{R}^2\right)\rightarrow\mathcal{L}_2\left(\mathbb{R}^2\right)$ and  $\boldsymbol R_y^{\{n,c\}}:\mathcal{H}^{11/6}\left(\mathbb{R}^2\right)\rightarrow\mathcal{L}_2\left(\mathbb{R}^2\right)$  defined by
\begin{align}\label{eq:3.6a} 
\left( \boldsymbol R_x^{\{n,c\}} \Phi\right) (x,y) &:=\mathcal{X}_{\Omega}(x,y)   \dfrac { 1 } { \pi } \int_{\Omega_y} \dfrac{ \sin [ \Phi (x',y) - \Phi (x,y) ] \times k^{\{n,c\}} (x'-x) } {x'-x} \ \mathrm{d}x', \\ \label{eq:3.6b}
\left( \boldsymbol R_y^{\{n,c\}} \Phi\right) (x,y) &:=\mathcal{X}_{\Omega}(x,y)   \dfrac { 1 } { \pi }\int_{\Omega_x} \dfrac{ \sin [ \Phi (x,y') - \Phi (x,y) ] \times k^{\{n,c\}}(y'-y) } {y'-y} \ \mathrm{d}y' \\ \notag
\end{align}
for the functions $k^{\{n,c\}}$ given as in Theorem \ref{pyr_def}.
\end{theorem}

\begin{proof}
See \cite{BuDa06,Shatokhina_PhDThesis,Veri04}.
\end{proof}

The operators $\boldsymbol P^{\{n,c\}}_x$ and $\boldsymbol P^{\{n,c\}}_y$ as well as $\boldsymbol R^{\{n,c\}}_x$ and $\boldsymbol R^{\{n,c\}}_y$ are constructed in the same way. Please note that due to the similar structure of the operators, results for one direction are immediately transferred to the second direction. \bigskip

In closed loop AO, the pyramid sensor measures already corrected (and thus small) wavefront aberrations.  A linearization of the pyramid WFS operators $\boldsymbol P^{\{n,c\},lin}_x$ and roof WFS operators $\boldsymbol R^{\{n,c,l\},lin}_x$ is obtained either by means of the Fr\'{e}chet derivative as shown in \cite{HuSha18_1} or by the replacement 
$$ \sin [ \Phi (x',y) - \Phi (x,y) ] \approx \Phi (x',y) - \Phi (x,y)  $$
which is according to \cite{BuDa06} valid in closed loop AO because of
$$\left|\Phi (x',y) - \Phi (x,y) \right| << 1 .$$

\begin{prop}\label{3.8b}
The linear approximations $\boldsymbol R_x^{\{n,c,l\},lin}:\mathcal{H}^{11/6}\left(\mathbb{R}^2\right)\rightarrow\mathcal{L}_2\left(\mathbb{R}^2\right)$ of the roof sensor operators in $x$-direction are given by
\begin{equation}\label{eq:3.81}
\left(\boldsymbol R_x^{\{n,c\},lin}\Phi\right)(x,y):=\mathcal{X}_{\Omega}(x,y)\dfrac{1}{\pi}\int_{\Omega_y}{\dfrac{ [\Phi(x',y)-\Phi(x,y)] \times k^{\{n,c\}} (x'-x) }{x'-x}\ \mathrm{d}x'}
\end{equation}
and $\boldsymbol R^{\{n,c,l\},lin}_y$ accordingly.
\end{prop}
 
\section{Non-linear Landweber iteration for solving the wavefront reconstruction problem} \label{chap_landweber} 
 
Almost all existing wavefront reconstruction approaches assume a linear relation between the unknown wavefront or mirror actuator commands and the given pyramid wavefront sensor measurements \cite{ShatHut_spie2018_overview}. Such reconstruction methods for pyramid sensors are, for instance, several variations of MVM (matrix-vector multiplication) approaches summarized in \cite{HuShaOb18},
 Fourier based methods \cite{QuPa10,Shat17_ao4elt5_clif,Shat17}, approaches based on the inversion of the finite Hilbert transform \cite{Hut17,Shatokhina_PhDThesis}, the fast P-CuReD algorithm \cite{Ros11,Shat13} or rapidly converging iterative reconstruction methods \cite{HuSha18_2,HuSha18_1} which were just recently developed. In \cite{HuSha18_2}, the authors already applied Landweber iteration for wavefront reconstruction from pyramid sensor data but based on the linearization of the model. The results contained therein will serve as comparison to a potential improvement using non-linear methods as discussed in one of the subsequent Sections. Additionally, in \cite{HuSha18_2} the authors reconstruct by only using the one-term assumption of the roof sensor model, i.e., they exclude the second term in Eq.~\eqref{eq:3.81} and consider variations of the finite Hilbert transform in their numerical implementations. In our approaches, we use the full roof sensor model, and hence may achieve additional gains in reconstruction performance.\bigskip

In this paper we propose two new non-linear reconstruction algorithms for pyramid sensor measurements. As non-linear reconstruction procedure, we use either Landweber or Landweber-Kaczmarz iteration. For these choices, we chiefly benefit from the stability and flexibility of the algorithms with regard to the optimization of numerical parameters. 
For an appropriate choice of the relaxation parameter, we even observe an accelerated convergence of the method. \bigskip

We first review the general scheme of the Landweber algorithm originally proposed in the 1950s by Louis Landweber \cite{Landweber51} for linear problems and emphasize that the theory for solving non-linear, ill-posed problems using Landweber iteration considered in this paper is mainly referred to \cite{Neub95,kaltenbacher2008iterative}.  \bigskip

The Landweber method is an iterative technique for minimizing the residual as the quadratic functional
\begin{equation}\label{eq:l.1}\notag
\left|\left|\boldsymbol P\Phi - s \right|\right|^2
\end{equation}
and is an appealing iterative alternative to Tikhonov regularization. 
The idea of Landweber iteration follows from the fixed point equation $\Phi_{k+1} = F\left(\Phi_k\right)$ with fixed point operator
\begin{equation}\label{eq:l.3}
\boldsymbol F(\Phi) := \Phi + \boldsymbol P'\left(\Phi\right)^*\left(s-\boldsymbol P\left(\Phi\right)\right)
\end{equation}
assuming that the pyramid sensor operator is differentiable. 
The iterative procedure is described by
\begin{equation}\label{eq:l.2}
\Phi_{k+1}^\delta = \Phi_k^\delta + \boldsymbol P'\left(\Phi_k^\delta\right)^*\left(s^\delta-\boldsymbol P\left(\Phi_k^\delta\right)\right) \qquad \qquad k=0,1,2,\dots 
\end{equation}
with perturbed data $s^\delta$ fulfilling $\left|\left|s-s^\delta \right|\right|<\delta$ and adjoint $\boldsymbol P'\left(\Phi\right)^*$ of the locally uniformly bounded Fr\'{e}chet derivative $\boldsymbol P'\left(\Phi\right)$ in $\Phi$. Even if $s^\delta$ does not belong to the range of $\boldsymbol P$, the Landweber iteration~\eqref{eq:l.2} is stable for a fixed number of iterations. The number of iterations acts as regularization in case of noisy data $s^\delta$. 

Since the pyramid sensor decouples into two directions, we will consider $x$- and $y$- direction as independent Landweber iterations and are interested in a reconstruction $\left[\Phi_x,\Phi_y\right]$ for every direction. The final reconstruction $\Phi$ is then calculated as the average of both directions. Due to the symmetry of the problem, both directions can be handled analogously, and hence are not examined individually here. 

The iterates start with an initial guess $\Phi_0$ which may include a priori knowledge of an exact solution $\Phi_*$.  We always assume $\Phi_0^\delta = \Phi_0$. In AO loops, reconstructions are needed for a number of time steps and the reconstruction of the previous time steps usually is a good choice for the initial guess of the current time step. Hence, for subsequent time steps, we consider a warm restart technique meaning that as initial guess at time step $t+1$ denoted by $\Phi_{0,t+1}$ we choose the reconstruction of the last time step $t$, i.e., $\Phi_{0,t+1} = \Phi_{rec,t}$. 

As regularization we use the discrepancy principle for regularization in case of noisy data: For an appropriate choice of $\tau > 0$ the iteration is stopped after $k_* = k_*(\delta,s^\delta)$ iteration steps fulfilling
\begin{equation}\label{eq:l.33}
 \left|\left|s^\delta - \boldsymbol P\left(\Phi_k^\delta\right)\right|\right| > \tau\delta, \qquad 0 \le k < k_* \qquad \qquad \text{and} \qquad \qquad \left|\left|s^\delta - \boldsymbol P\left(\Phi_{k_*}^\delta\right)\right|\right| \le \tau\delta.
\end{equation}
For non-linear inverse problems, iteration procedures as, e.g., \eqref{eq:l.2} will in general not converge globally to a solution of the non-linear operator equation~\eqref{eq:n.1}. The convergence theory is, in general, based on the assumption that the fixed point operator $\boldsymbol F$ is a nonexpansive operator, i.e.,$$\left| \left|\boldsymbol F\left(\Phi\right)-\boldsymbol F\left(\tilde{\Phi}\right)\right|\right| \le \left|\left|\Phi-\tilde{\Phi}\right|\right|, \qquad \qquad \Phi,\tilde{\Phi} \in \mathcal{D}\left(\boldsymbol F\right).$$ Iterative methods for approximating fixed points of nonexpansive operators have been considered, for instance, in \cite{Baku89, Baku2004, Brow67, Hal67}. As in many applications, it is difficult to verify analytically whether the fixed point operator $\boldsymbol F$ is nonexpansive for the pyramid sensor. For that reason, the nonexpansivity of the fixed point operator is often replaced by properties that guarantee at least the local convergence of the iteration method and are easier to check.
To obtain local convergence we assume the pyramid sensor operator equation~\eqref{eq:n.1} to be scaled according to
\begin{equation}\label{eq:l.4}
\left|\left|\boldsymbol P'\left(\Phi\right)\right|\right| \le 1, \qquad \quad \qquad \Phi \in \mathcal{B}_{2\rho}\left(\Phi_0\right)\subset \mathcal{D}\left(\boldsymbol P\right)
\end{equation}
for a ball $\mathcal{B}_{2\rho}\left(\Phi_0\right)$ of radius $2\rho$ around the initial guess $\Phi_0$.

The second condition needed to ensure local convergence in $\mathcal{B}_\rho\left(\Phi_0\right)$ to a solution of Eq.~\eqref{eq:n.1} reads as
\begin{equation}\label{eq:l.5}
\left|\left|\boldsymbol P\left(\Phi\right)-\boldsymbol P\left(\tilde{\Phi}\right)-\boldsymbol P'\left(\Phi\right)\left(\Phi-\tilde{\Phi}\right) \right|\right| \le \eta \left|\left|\boldsymbol P\left(\Phi\right)-\boldsymbol P\left(\tilde{\Phi}\right)\right|\right|,  
\end{equation}
for $\eta < \dfrac{1}{2}$, $\Phi,\tilde{\Phi} \in \mathcal{B}_{2\rho}\left(\Phi_0\right) \subset \mathcal{D}\left(\boldsymbol P\right)$. Both conditions guarantee that the Landweber iteration is well-defined as all iterates $\Phi_k^\delta, 0 \le k \le k_*$ remain elements of $\mathcal{D}\left(\boldsymbol P\right)$ if we employ the discrepancy principle. \bigskip

Under the above mentioned conditions~\eqref{eq:l.4}~-~\eqref{eq:l.5} and the discrepancy principle~\eqref{eq:l.33} we obtain the following convergence results.

\begin{theorem}[\cite{kaltenbacher2008iterative}, Theorem 2.4 \& Theorem 2.6]
If we assume that Eq.~\eqref{eq:n.1} is solvable in $\mathcal{B}_\rho\left(\Phi_0\right)$ and that the conditions~\eqref{eq:l.4}~-~\eqref{eq:l.5} hold, the non-linear Landweber iteration converges to a solution of $\boldsymbol P \Phi=s$ in case of exact data $s$. 

If $k_*\left(\delta,s^\delta\right)$ defines a stopping index according to the discrepancy principle~\eqref{eq:l.33} for
\begin{equation}\label{eq:l.55}
\tau > 2\dfrac{1+\eta}{1-2\eta} > 2
\end{equation}
with $\eta$ as in~\eqref{eq:l.5}, the Landweber iterates $x_{k_*}^\delta$ converge to a solution of Eq.~\eqref{eq:n.1}.
 
If $\mathcal{N}\left(\boldsymbol P'\left(\Phi^\dagger\right)\right) \subset \mathcal{N}\left(\boldsymbol P'\left(\Phi\right)\right)$ for all $\Phi \in \mathcal{B}_\rho\left(\Phi^\dagger\right)$, we obtain convergence of $\Phi_k$ and, respectively, of $\Phi_{k_*}^\delta$ to the $\Phi_0$-minimum-norm solution $\Phi^\dagger$ as $k \rightarrow \infty$ and $\delta \rightarrow 0$. 
\end{theorem}

In order to ensure condition~\eqref{eq:l.4} we introduce a relaxation parameter $\omega$, which is chosen such that $$\omega\left|\left|\boldsymbol P'\left(\Phi^\dagger\right)\right|\right|^2 \le 1,$$ and consider the iteration scheme
\begin{equation}\label{eq:l.6}
\Phi_{k+1}^\delta = \Phi_k^\delta + \omega \boldsymbol P'\left(\Phi_k^\delta\right)^*\left(s^\delta-\boldsymbol P\left(\Phi_k^\delta\right)\right) \qquad \qquad k=0,1,2,\dots 
\end{equation}
as a modification of Eq.~\eqref{eq:l.2}.

\subsection{Non-linear Landweber algorithm applied to pyramid sensors}

For the first algorithm, named \textit{LIPS (Landweber Iteration for Pyramid Sensors)}, we apply Landweber iteration in $x$-direction and a second one independently in $y$-direction. Thus, we obtain two reconstructions $\Phi=\left[\Phi_x,\Phi_y\right]$ of the incoming phase and average them at the end. As we observed in numerical simulations, the reconstruction quality exhibits most improvements in the first few Landweber iterations and little in subsequent iterations, we fix the number of iterations $K$ in advance instead of using the discrepancy principle~\eqref{eq:l.33} in order to reduce the computational load of the method. \bigskip

The non-linear Landweber iteration modified for wavefront reconstruction from pyramid sensor data using a warm restart of the system reads as: 
\begin{table}[h]
\renewcommand{\arraystretch}{1.2}
\begin{tabular}{lcr}
\hline
\textbf{Algorithm 1} non-linear Landweber Iteration for Pyramid Sensors (LIPS)  \\
\hline
choose initial guess $\Phi_{K,0}$, set relaxation parameter $\omega$ \\
for $t=1,\dots T$ do \\
\quad $\Phi_{0,t} = \Phi_{K,t-1}$ \\
\quad for $i=1,\dots K$ do \\
\quad \quad $ \Phi_{i,t} = \Phi_{i-1,t} + \omega \boldsymbol P'\left(\Phi_{i-1,t}\right)^*\left(s_t-\boldsymbol P\left(\Phi_{i-1,t}\right)\right)$\\ 
\quad endfor \\
endfor\\
$\Phi_{K,T} =  \left(\Phi_{x,K,T}+\Phi_{y,K,T}\right)/2$ \\
\hline
\end{tabular} 
\end{table}

\subsection{Non-linear Landweber-Kaczmarz algorithm applied to pyramid sensors}

The structure of the operator equation~\eqref{eq:n.1} consisting of two equations (one in $x$- and one in $y$-direction) leads us to further consider applying Landweber-Kaczmarz iteration instead of simply averaging the two obtained reconstructions. In contrast to the two reconstructions $\left[\Phi_x,\Phi_y\right]$, one for each direction, we are now only interested in one reconstruction for both directions, again denoted by $\Phi$. Advantages of Kaczmarz strategies for wavefront reconstruction using pyramid sensors have already been analyzed in~\cite{HuSha18_2}.

The principal idea of Kaczmarz's method \cite{Bank82,McCormick77,McCormick1975,Meyn1983,Natterer2001} can be used in combination with any iterative procedure. In image reconstruction, the method, which is also known as algebraic reconstruction technique (ART), was used in \cite{Herman80} for the first time. \bigskip

We apply the Kaczmarz-type method in combination with Landweber iteration to the problem of wavefront reconstruction from pyramid sensor data and name the new method \textit{KLIPS (Kaczmarz-Landweber Iteration for Pyramid Sensors)}. \bigskip

The non-linear Landweber-Kaczmarz iteration modified for wavefront reconstruction from pyramid sensor data using a warm restart of the system is described by: 

\begin{table}[h]
\renewcommand{\arraystretch}{1.2}
\begin{tabular}{lcr}
\hline
\textbf{Algorithm 2} non-linear Kaczmarz-Landweber Iteration for Pyramid Sensors (KLIPS)  \\
\hline
choose initial guess $\Phi_{K,0}$, set relaxation parameter $\omega_x$, $\omega_y$ \\
for $t=1,\dots T$ do \\
\quad $\Phi_{0,t} = \Phi_{K,t-1}$ \\
\quad for $i=1,\dots K$ do \\
\quad \quad $ \Phi_{i,t,1} = \Phi_{i-1,t} + \omega_{x} \boldsymbol P_x'\left(\Phi_{i-1,t}\right)^*\left(s_{x,t}-\boldsymbol P_x\left(\Phi_{i-1,t}\right)\right)$\\ 
\quad \quad $ \Phi_{i,t,2} = \Phi_{i,t,1} + \omega_{y} \boldsymbol P_y'\left(\Phi_{i,t,1}\right)^*\left(s_{y,t}-\boldsymbol P_y\left(\Phi_{i,t,1}\right)\right)$\\ 
\quad \quad $ \Phi_{i,t} = \Phi_{i,t,2}$ \\
\quad endfor \\
endfor\\
\hline
\end{tabular}
\end{table}

Instead of averaging the reconstructions in the last step as in the previously described Landweber iteration for pyramid sensors, we now apply the Landweber iteration steps cyclically. For convergence and stability proofs of this method \cite{Cez08,Kow02}, one basically has to enforce the same conditions on every involved operator as for Landweber iteration. Since for the LIPS we have already considered the two directions as completely independent Landweber iterations, no further conditions necessary for the convergence and stability of the KLIPS have to be shown. As for the LIPS, an appealing alternative to the discrepancy principle in AO loops for pyramid sensors is to fix the number of iterates in advance in order to avoid time-consuming computations which do not deliver high quality improvements. In case we want to use the discrepancy principle we modify~\eqref{eq:l.6} by
\begin{equation}\label{eq:l.7}
\Phi_{k+1}^\delta = \Phi_k^\delta + \sigma_{j,k} \omega_j \boldsymbol P_j'\left(\Phi_k^\delta\right)^*\left(s_j^\delta-\boldsymbol P_j\left(\Phi_k^\delta\right)\right) \qquad \qquad k=0,1,2,\dots 
\end{equation}
and use
\begin{align}\label{l.8}
\sigma_{j,k} := \begin{cases} 1, &\mathrm{if} \ \tau\delta < \left|\left|s_j^\delta - \boldsymbol P_j\left(\Phi_k^\delta\right) \right|\right|, \\
0, \qquad \qquad &\mathrm{else}, \end{cases}
\end{align}
for $j=x,y$ indicating the direction, $k$ denoting the iteration step and $\tau$ chosen according to~\eqref{eq:l.55}. The iteration procedure is stopped in case of a stagnation over one full cycle of iterates. 

Note that both given algorithms have the warm restart technique incorporated but can also be used without a warm restart.

\section{Fr\'{e}chet derivatives and corresponding adjoints of the roof sensor operators} \label{chap_frechet}
 
For the application of the above proposed algorithms LIPS and KLIPS we need to calculate the Fr\'{e}chet derivatives and the corresponding adjoint operators. At this point, we can either use the full pyramid sensor model or the roof sensor as a simplification. Here, we consider the latter, i.e., we plug in the operators $\boldsymbol R^{\{n,c\}}$ introduced in Eq.~\eqref{eq:3.6a}~-~\eqref{eq:3.6b} instead of the operator $\boldsymbol P$ in Eq.~\eqref{eq:l.2} for the reconstruction, but still use the full pyramid wavefront sensor data $s$. The roof sensor models have already been taken as basis for several developed algorithms most of which are adapted from a linearization of the roof sensor operators. There even exist reconstruction methods that reduce the full pyramid sensor model to a one-term assumption meaning that the pyramid sensor model is approximated as the finite Hilbert transform of the incoming phase \cite{ShatHut_spie2018_overview}. The high reconstruction performance delivered by these algorithms (even though they are based on simplifications) is confirmed in, e.g., \cite{Hut17,HuSha18_2,ShatHut_spie2018_overview,Shatokhina_PhDThesis,Shat17,Shat13}. 

\begin{prop}
The Fr$\acute{e}$chet derivatives $\left(\boldsymbol R^{\{n,c\}}\right)'\left(\Phi\right) \in \mathcal{L}\left(\mathcal{H}^{11/6}, \mathcal{L}_2\right)$ of the roof sensor operators $\boldsymbol R^{\{n,c\}}$ at $\Phi \in \mathcal{D}\left(\boldsymbol R^{\{n,c\}}\right)$ are given by 
{\small\begin{equation*}
\left(\left(\boldsymbol R_x^{\{n,c\}}\right)'(\Phi)\ \psi\right)\left(x,y\right)
=\mathcal{X}_{\Omega}\left(x,y\right)\left[\dfrac{1}{\pi} \int_{\Omega_y}{\dfrac{\cos{\left[\Phi(x',y)-\Phi(x,y)\right]}\times k^{\{n,c\}}\left(x'-x\right)\left[\psi(x',y)-\psi(x,y)\right]}{x'-x}\ \mathrm{d}x'}\right]
\end{equation*}
}and
{\small\begin{equation*}
\left(\left(\boldsymbol R_y^{\{n,c\}}\right)'(\Phi)\ \psi\right)\left(x,y\right)
=\mathcal{X}_{\Omega}\left(x,y\right)\left[\dfrac{1}{\pi}  \int_{\Omega_y}{\dfrac{\cos{\left[\Phi(x,y')-\Phi(x,y)\right]}\times k^{\{n,c\}}\left(y'-y\right)\left[\psi(x,y')-\psi(x,y)\right]}{y'-y}\ \mathrm{d}y'}\right].
\end{equation*}}
\end{prop}
\begin{proof}
The Fr\'{e}chet derivatives of the operators representing the roof wavefront sensor were already derived in~\cite[Theorem~5]{HuSha18_1}. 
\end{proof}

The calculation of the adjoint operators depend on the used spaces and the accordant inner product. As we consider $\left( \boldsymbol R^{\{n,c\}}_x\right)'\left(\Phi\right):\mathcal{H}^{11/6}\rightarrow \mathcal{L}_2$, the corresponding adjoint maps from $\mathcal{L}_2$ into $\mathcal{H}^{11/6}$. 
In order to calculate the adjoints for the operators defined from $\mathcal{H}^{11/6}$ into $\mathcal{L}_2$ we consider the embedding operator introduced as 
\begin{equation*}
i_s: \mathcal{H}^{11/6} \rightarrow \mathcal{L}_2
\end{equation*}
and derive the corresponding adjoint operators $\left(\left( \boldsymbol R^{\{n,c\}}_x\right)'\left(\Phi\right)\right)^*:\mathcal{L}_2\rightarrow \mathcal{H}^{11/6}$ in $\Phi$ according to \cite{RaTesch04} by $$\left(\left( \boldsymbol R^{\{n,c\}}_x\right)'\left(\Phi\right)\right)^* = i_s^*\left(\left( \tilde{\boldsymbol R}^{\{n,c\}}_x\right)'\left(\Phi\right)\right)^*$$ for $\left(\left( \tilde{\boldsymbol R}^{\{n,c\}}_x\right)'\left(\Phi\right)\right)^*: \mathcal{L}_2 \rightarrow \mathcal{L}_2$. Hence, it is sufficient to derive the adjoint operators $\left(\left(\tilde{ \boldsymbol R}^{\{n,c\}}_x\right)'\left(\Phi\right)\right)^*$ with respect to the inner product in $\mathcal{L}_2\left(\mathbb{R}^2\right)$. For simplicity of notation, we use $\left(\left( \boldsymbol R^{\{n,c\}}_x\right)'\left(\Phi\right)\right)^*$ for $\left(\left( \tilde{\boldsymbol R}^{\{n,c\}}_x\right)'\left(\Phi\right)\right)^*$ in the following. Details on the implementation of the above considerations can be found in \cite{RaTesch04}.

\begin{prop}\label{roof_frechet_adjoint}
The adjoint operators $\left(\left(\boldsymbol R^{\{n,c\}}\right)'(\Phi)\right)^*:\mathcal{L}_2\left(\mathbb{R}^2\right)\rightarrow \mathcal{L}_2\left(\mathbb{R}^2\right)$ of the roof sensor's Fr\'{e}chet derivatives at $\Phi$ are represented by
{\footnotesize\begin{equation}\label{frech_def_a}
\left(\left(\boldsymbol R_x^{\{n,c\}}\right)'(\Phi)\right)^* \psi\left(x,y\right)
=-\mathcal{X}_{\Omega}\left(x,y\right)\dfrac{1}{\pi}  \int_{\Omega_y}{\dfrac{\cos{\left[\Phi(x',y)-\Phi(x,y)\right]}\times k^{\{n,c\}}\left(x'-x\right)\left[\psi(x',y)+\psi(x,y)\right]}{x'-x}\ \mathrm{d}x'}
\end{equation}
}and
{\footnotesize\begin{equation}\label{frech_def_b}
\left(\left(\boldsymbol R_y^{\{n,c\}}\right)'(\Phi)\right)^* \psi\left(x,y\right)
=-\mathcal{X}_{\Omega}\left(x,y\right)\dfrac{1}{\pi}   \int_{\Omega_x}{\dfrac{\cos{\left[\Phi(x,y')-\Phi(x,y)\right]}\times k^{\{n,c\}}\left(y'-y\right)\left[\psi(x,y')+\psi(x,y)\right]}{y'-y}\ \mathrm{d}y'}.
\end{equation}}
\end{prop}
\begin{proof}
The proof is given in the Appendix.
\end{proof}

\section{Numerical implementation and complexity} \label{chap_complexity}
Let us now summarize some numerical aspects as well as the computational complexities of the proposed methods. First, we choose an adequate representation of the incoming wavefront and the measurements for an $N\times N$ pyramid sensor.
We assume the two dimensional wavefront to be given as linear combination of characteristic functions $\mathcal{X}_{\Omega_{ij}}\left(x,y\right)$ of the subapertures $\left(\Omega_{ij}\right)_{i,j=1}^N$.
For the annular telescope aperture $\Omega = \Omega_y\times \Omega_x$, the disjoint areas $\Omega_{ij}$ are chosen such that
\begin{equation}\label{eq:im.3}
\Omega=\bigcup\limits_{i,j=1}^N \Omega_{ij} \qquad \qquad \text{and} \qquad \qquad \Omega_{ij} \cap \Omega_{mn} = \emptyset\qquad \text{for} \ i\neq m \land j\neq n
\end{equation}
as well as
$$\Omega_n = \bigcup\limits_{i=1}^N \Omega_{in},$$ i.e., the term $\Omega_n$ indicates the row of the symmetric telescope pupil $\Omega$ located at $y$-position $n$. In order to describe an annular telescope aperture instead of a squared one, we assign all areas $\Omega_{ij}$ being located outside the annular aperture to the empty set in~\eqref{eq:im.3} just for simplicity of notation. In reality, we have different numbers of subapertures in every row and column of the aperture. For the numerical implementation, we only consider those subapertures being located on the aperture to save computation time. The incoming phase is represented by
\begin{equation}\label{eq:im.1}
\Phi\left(x,y\right) = \sum_{i,j=1}^N \Phi_{ij}\mathcal{X}_{\Omega_{ij}}\left(x,y\right)
\end{equation}
with coefficients $\Phi_{ij} \in \mathbb{R}, 1 \le i,j\le N$. \\
The pyramid sensor delivers data on every subaperture -- more precisely, one measurement in $x$-direction and one measurement in $y$-direction for every subaperture. For these we choose the same representation as for the incoming phase, i.e.,
\begin{equation}\label{eq:im.2}
s\left(x,y\right)= \sum_{i,j=1}^N s_{ij}\mathcal{X}_{\Omega_{ij}}\left(x,y\right),
\end{equation}
where $s_{ij} \in \mathbb{R}, 1 \le i,j \le N$ denotes the coefficients for measurements $s = s_x$ or $s=s_y$ respectively. The points $\left(x_m,y_n\right)$ where data are assumed to be given can, for instance, be chosen as the middle points of every subaperture of the symmetric pupil, i.e.,
$$x_m = \dfrac{-D-d}{2}+m\times d, \qquad y_n=\dfrac{-D-d}{2}+n\times d \qquad \qquad \text{for}\  m,n=1,\dots,N.$$
Here, $d=D/N$ denotes the subaperture size.

\begin{prop}\label{prop_discretization}
Using the representations~\eqref{eq:im.1} for the incoming phase $\Phi$ and the direction $\psi$, the evaluation of the operators $\boldsymbol R^{\{n,c\}}$ and $\left(\left(\boldsymbol R^{\{n,c\}}\right) ' \left(\Phi\right)\right)^*$ defined according to Eq.~\eqref{eq:3.6a}~-~\eqref{eq:3.6b} and \eqref{frech_def_a}~-~\eqref{frech_def_b} in the middle point $\left(x_m,y_n\right)$ of the subaperture $\Omega_{mn}\subset \Omega$ for $1 \le m,n \le N$ can be represented by
\begin{equation*}
\left( \boldsymbol R_x^{\{n,c\}} \Phi\right)_{m,n} = \dfrac { 1 } { \pi }\ \sum_{\substack{i=1\\i \neq m}}^N \sin\left[\Phi_{in}-\Phi_{mn}\right] \alpha_{in}^{\{n,c\}}\left(x_m\right)
\end{equation*}
and
\begin{equation*}
\left(\left(\left(\boldsymbol R_x^{\{n,c\}}\right)'(\Phi)\right)^* \psi\right)_{mn} =-\dfrac{1}{\pi}\sum\limits_{i=1}^N \cos{\left[\Phi_{in}-\Phi_{mn}\right]}\left[\psi_{in}+\psi_{mn}\right] \alpha_{in}^{\{n,c\}}\left(x_m\right)
\end{equation*}
with
\begin{equation}\label{eq:im.4}
\begin{split}
\alpha_{in}^{\{n,c\}}\left(x_m\right):= \begin{cases} \ \int_{\Omega_{in}} \dfrac{  k^{\{n,c\}} (x'-x_m) } {x'-x_m} \ \mathrm{d}x',   &\mathrm{for} \ i\neq m, \\
\ p.v.\ \int_{\Omega_{mn}} \dfrac{  k^{\{n,c\}} (x'-x_m) } {x'-x_m} \ \mathrm{d}x', \qquad \qquad  &\mathrm{for} \ i=m. \end{cases}
\end{split}
\end{equation}
\end{prop}

\begin{proof}
The proof is given in the Appendix.
\end{proof}

Please note that the principal value meaning only has to be utilized when computing the function values $\alpha_{mn}^{\{n,c\}}\left(x_m\right)$. 

The benefit of these representations is that we can precompute the functions $\alpha_{in}^{\{n,c\}}$ offline which significantly reduces the computational load of the algorithms. Almost exclusively sine and cosine evaluations as well as multiplications with the precomputations have to be performed online.

\subsection{Computational complexity}

For the evaluation of the numerical effort of the Landweber and Landweber-Kaczmarz iteration for pyramid sensors we only consider the complexity of the operations that have to be performed online. We exclude the pre-calculations $\alpha^{\{n,c\}}$ in Eq.~\eqref{eq:im.4} needed for the application of the operators $\boldsymbol R$ and $\left(\boldsymbol R'\right)^*$ from our considerations. As before, $N$ indicates the number of subapertures in one direction. Since we send deformable mirror actuator commands to the AO system, the number of active actuators $N_a \sim N^2$ denotes the number of unknowns to be found. Please note that the proposed algorithms provide the reconstructed wavefront. An additional step of transforming the reconstructed phase into mirror actuator commands must be applied at the end. The effort of this projection step is not considered in the following. 

An overview of the number of FLOPS for both algorithms is provided in Table~\ref{table:PKI_FLOPs}. The post loop step necessary in case of the LIPS consists of finding the average of the two resulting reconstructions, which requires one summation of two $N\times N$ matrices as well as one division by a scalar. Altogether, this step is summed up to $2 N^2$ operations. \bigskip

\begin{table}
\renewcommand{\arraystretch}{1.2}
\begin{center}
\caption{\textbf{Number of FLOPS to be performed online in the LIPS and KLIPS method.}}
\begin{small}
\begin{tabular}{ l| l   l  } 
\hline
& \textbf{operation} & \textbf{\# of FLOPS} \\
 \hline
loop & $\boldsymbol P_x \Phi_x$ & $2N^3$  \\
& $s_x-\boldsymbol P_x\Phi_x$ & $N^2$ \\
& $\boldsymbol P_x'\left(\Phi_x\right)^*\left(s_x-\boldsymbol P_x\Phi_x\right)$ & $2N^3+2N^2$ \\
&$\omega\boldsymbol P_x'\left(\Phi_x\right)^*\left(s_x-\boldsymbol P_x\Phi_x\right)$ & $N^2$ \\
&$\Phi_x+\omega\boldsymbol P_x'\left(\Phi\right)^*\left(s_x-\boldsymbol P_x\Phi_x\right)$ & $N^2$ \\
post loop step & $\Phi = \tfrac{1}{2}\left(\Phi_x+\Phi_y\right)$ & $2N^2$ \\
\hline
\end{tabular}  \\ The post loop step is only necessary for the LIPS and omitted for the KLIPS.\label{table:PKI_FLOPs}
\end{small}
\end{center}
\end{table}

Since for both algorithms we perform the mentioned operations twice (in $x$- and in $y$-direction), for $K$ iterations we end up with $$C_{LIPS}(N; K)=\left(8N^3+10N^2\right)\times K + 2N^2 \qquad  \text{FLOPS}$$ for the application of the Landweber iteration (Algorithm $1$) having the additional step of averaging and $$C_{KLIPS}(N; K)=\left(8N^3+10N^2\right)\times K \qquad  \text{FLOPS}$$ for the non-linear Landweber-Kaczmarz approach (Algorithm $2$). Thus, both algorithms have a computational effort of approximately $\mathcal{O}\left(N_a^{3/2}\right)$.

\section{End-to-end simulation results}   \label{chap_numerics}

 To analyze the performance quality of the presented non-linear algorithms we simulate an instance of the ELT. We evaluate the reconstruction quality in a closed loop setting and compare the results with those of selected linear algorithms for wavefront reconstruction.  \bigskip

The presented closed loop results correspond to a METIS-like case \cite{METIS_spie_2016}. METIS will be a first light instrument of the ELT having a SCAO (Single Conjugate Adaptive Optics) module incorporated. For SCAO, the main interest is in an optimal image quality at the center of the field of view. There, one natural guide star is located and used as reference source for correcting aberrations induced by atmospheric turbulence. The object of interest has to be located in a close vicinity of the guide star. This simultaneously implies a major drawback of the most simple AO module. Since the light of both guide star and observed object travels approximately parallel to the optical axis, the occurring atmospheric perturbations simply add up. Hence, for the system it is sufficient to only have one deformable mirror and one wavefront sensor both conjugated to the ground layer. No atmospheric tomography is needed for SCAO. It is not necessary to have a deeper knowledge about height or thickness of turbulent layers in the wavefront reconstruction process. Only the sum of all phase shifts along the optical axis is measured for a subsequent control of the deformable mirror.\bigskip

The performance of the proposed methods is demonstrated with the official end-to-end simulation tool of the European Southern Observatory called Octopus \cite{LVK06,LeLouarn_OCTOPUS_04}. As primary mirror diameter of the ELT we simulate $39$ meters of which only the inner $37$ meters are used for METIS. Roughly $30$~$\%$ of the light gathering device are obstructed by a secondary mirror. Octopus generates a von Karman realization of the atmosphere having $35$ frozen layers at heights between $30$~m and $26.5$~km. The Fried parameter is $r_0=15.7$~cm and the outer scale $L_0=25$~m. The data are simulated by Octopus using the built-in model of a pyramid wavefront sensor with or without modulation on $74\times 74$ subapertures at a sensing wavelength of $\lambda = 2.2$~$\mu$m (K-band). The pyramid sensor measurements are read out $500$ or $1000$ times per second. Due to photon and read-out noise, we obtain perturbed pyramid sensor measurements. The deformable mirror geometry corresponds to the currently in Octopus implemented M4 geometry that is planned for the ELT. \bigskip

The reconstruction quality is quantified in terms of the short-exposure (SE) and the long-exposure (LE) Strehl ratio. This quality metric commonly used in astronomical communities is defined as the ratio of the peak aberrated image intensity from a point source compared to the maximum attainable intensity for an ideal optical system limited only by diffraction over the telescope aperture  \cite{Kechnie2016,Roddier,RoWe96}. Thus, the maximum Strehl ratio is equal to $1$. The SE Strehl ratio $S$ can be estimated by the Mar\'{e}chal's approximation as $$S\left(\Phi\right) \approx exp\left[-\left(2\pi\times \sigma\left(\Phi\right)/\lambda_{science}\right)^2\right].$$ 
The term $\lambda_{science}$ denotes the observing wavelength and $\sigma\left(\Phi\right)$ the root-mean-square deviation of the wavefront $\Phi$. The average on-axis SE Strehl ratio over the whole observing time is related by the LE Strehl ratio. The observing wavelength for the results presented in the following corresponds to $\lambda_{science}=2.2$~$\mu$m. Table~\ref{n_table:1} provides an overview of the simulation parameters. The parameters are chosen such that they conform to those chosen in \cite{HuSha18_2} for reasons of a direct comparison of the in this paper proposed non-linear methods to linear Landweber and Landweber-Kaczmarz iteration introduced therein. That is as well the reason for the choice of the observing wavelength  which is, according to the METIS specifications, not in the science range but useful for our analysis purposes. \bigskip


Numerical tests in closed loop AO are performed for a range of photon fluxes between $50$ and $10000$ photons per subapertures per frame. A simple integrator is used for the temporal control of the system. The gain is optimized manually and for the non-modulated sensor the same for all test cases, which underlines the stability of the algorithms with respect to parameter tuning. For the modulated sensor, the gain was adjusted in the simulations with $50$ and $100$ ph/subap/frame. We found that it is advantageous to have a frequency dependent loop gain correcting with main emphasis on low-order modes at the beginning of a closed loop simulation. \bigskip 

\begin{table}
\renewcommand{\arraystretch}{1.2}
\begin{center}
\caption{\textbf{Overview of simulation parameters.}}\label{n_table:1}
\begin{small}
\begin{tabular}{ l   l  } 
 \hline
 \textbf{SIMULATION PARAMETERS} &   \textbf{METIS-like simulation}  \\
  \hline
 telescope diameter & $37$~m    \\
 central obstruction & $30\%$  \\
science target & on-axis (SCAO)   \\ 
WFS & PWFS   \\
sensing band $\lambda$ & K ($2.2$~$\mu$m)   \\
evaluation band $\lambda_{science}$ & K ($2.2$~$\mu$m)  \\
modulation  & $\left[0,4\right]$ $\lambda/D$  \\
controller &  integrator  \\
atmospheric model & von Karman   \\
number of simulated layers & $35$   \\
outer scale $L_0$ & $25$~m   \\
atmosphere & median   \\
Fried radius $r_0$ at $\lambda = 500$~nm & $0.157$~m   \\
number of subapertures & $74 \times 74$  \\
number of active subapertures & $3912$ out of $5476$  \\
linear size of simulation grid & $740$ pixels \\
DM geometry & ELT M4 model \\
DM delay & $1$  \\
frame rate &  $\left[1000,500\right]$~Hz   \\ 
detector read-out noise & $1$ electron/pixel \\
background flux & $0.000321$ photons/pixel/frame \\
photon flux &   $[50,100,1000,10000]$ ph/subap/frame  \\
time steps &  $500$  \\ 
 \hline
\end{tabular}
\end{small}
\end{center}
\end{table}

\begin{figure}[htbp]
\centering
\includegraphics[scale=0.75]{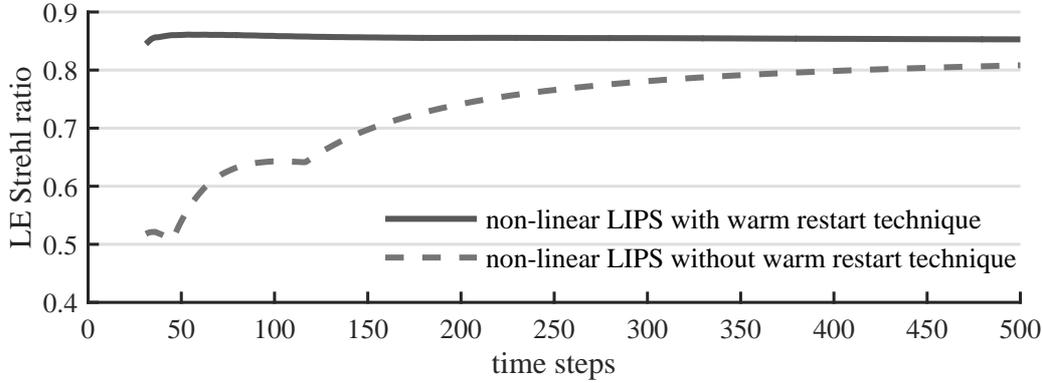}
\caption{Warm restart effect for the non-modulated sensor: Without using the warm restart technique (dotted line) the reconstructor suffers from slower convergence compared to the method with a warm restart (solid line). The results correspond to the LIPS approach for $10000$ ph/subap/frame.}
\label{fig:warm_restart}
\end{figure}

The number of Landweber or Landweber-Kaczmarz iterations is limited to $5$. This allows for gains in savings with respect to the computational load at the expense of loosing the rather minor quality improvements delivered by further iterations. As already discussed, we employ the warm restart technique. The initial guess for the first time step is chosen as zero. Then, for time steps $t>0$ the initial value is assigned the reconstruction of the last step. Omitting the warm restart but performing the same number of iterations per time step, we observe that the convergence is severely hampered as shown in Fig.~\ref{fig:warm_restart} for the non-modulated sensor. However, the warm restart technique was not effective for the modulated sensor. \bigskip

\begin{table}
\renewcommand{\arraystretch}{1.2}
\begin{center}
\caption{\textbf{LE Strehl ratios in the K-band for PWFS without modulation.}}\label{table:numerical_results}
\begin{small}
\begin{tabular}{ l c c c c } 
 \hline
 \textbf{photon flux} & \textbf{non-lin. LIPS}  & \textbf{non-lin. KLIPS} & \textbf{lin. LIPS \cite{HuSha18_2}} & \textbf{lin. KLIPS \cite{HuSha18_2}} \\
 \hline
50 &  0.8520 & 0.8517 & 0.8332 & 0.8371 \\
100 & 0.8534 & 0.8534& 0.8384 & 0.8415  \\
1000 & 0.8530  & 0.8531 & 0.8395 &  0.8420 \\
10000 & 0.8529  & 0.8530 & 0.8396 & 0.8419  \\
 \hline
\end{tabular} 
\end{small}
\end{center}
\end{table}

\begin{figure}[htbp]
\centering
\includegraphics[scale=0.65]{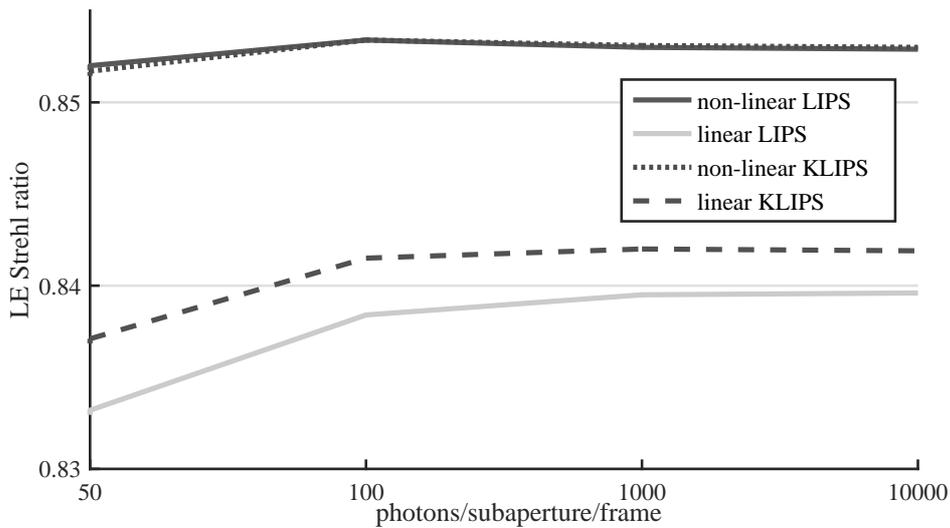}
\caption{Comparison of non-linear and linear methods for the non-modulated sensor: The non-linear algorithms outperform the linear approaches. While for the non-linear implementations there is almost no difference between the Landweber (solid lines) and the Landweber-Kaczmarz method (dashed lines), the Kaczmarz version of the linear reconstructors gives a higher improvement in quality.}
\label{fig:comparison_linear}
\end{figure}

\begin{table}
\renewcommand{\arraystretch}{1.2}
\begin{center}
\caption{\textbf{LE Strehl ratios in the K-band for PWFS with modulation.}}\label{table:numerical_results4}
\begin{small}
\begin{tabular}{ l  c c c c } 
 \hline
 \textbf{photon flux} & \textbf{non-lin. LIPS}  & \textbf{non-lin. KLIPS} & \textbf{lin. LIPS \cite{HuSha18_2}} & \textbf{lin. KLIPS \cite{HuSha18_2}} \\
 \hline
50 &  0.8155& 0.8145 & 0.8427 & 0.8439 \\
100 & 0.8253 & 0.8201 & 0.8517 & 0.8510  \\
1000 & 0.8333 & 0.8248 & 0.8590 &  0.8562 \\
10000 & 0.8342 & 0.8260 & 0.8595 & 0.8577  \\
 \hline
\end{tabular}
\end{small}
\end{center}
\end{table}

Since we observe a gain in performance by choosing a finer discretization, we divide every subaperture additionally into $3$ pixel areas. Thus, on cost of computation time, the quality is increased by about $0.034$ in terms of the LE Strehl ratio after $500$ time steps for the KLIPS in the high flux case, i.e., from $0.8193$ to $0.8530$ for the pyramid sensor without modulation. Additionally, we experience a faster convergence to higher Strehl ratios in case of a finer discretization.  \bigskip

As summarized in the first two columns of Table~\ref{table:numerical_results} for the non-modulated sensor and of Table~\ref{table:numerical_results4} for the modulated sensor, both proposed algorithms provide stable wavefront reconstruction with comparable results among each other in terms of LE Strehl ratios. 

 We obtain higher reconstruction quality with the non-modulated sensor compared to the modulated sensor in case we are using the proposed non-linear wavefront reconstruction methods. This may come from a higher sensitivity of the non-modulated sensor. Additionally, we experienced that for the modulated PWFS the non-linear algorithms are more sensitive with respect to loop gain and step size choices.\bigskip

\subsection{Comparison to the linear versions of the algorithms}

The sensor without modulation suffers from higher non-linearity effects compared to a sensor having an adequate modulation. In contrast, the expense of the improved linearity range of the modulated sensor is reduced sensitivity of the device \cite{Clare_2004,Fauv16,Fauv17,RaFa99,Veri04}. This property of the sensor gives us another reason why we are highly interested in an extension of the regime in which the non-modulated pyramid sensor effectually operates.

A comparison of the reconstruction quality using non-linear LIPS and non-linear KLIPS and their linear versions is given in Fig.~\ref{fig:comparison_linear} as well as Table~\ref{table:numerical_results} for the non-modulated sensor. For this sensor type, we obtain higher reconstruction quality with the non-linear processes. Additionally, the Kaczmarz-type methods outperform the standalone Landweber iterations in most of the cases. These conclusions cannot be directly transferred to the modulated sensor. The results for the modulated sensor in Fig.~\ref{fig:comparison_linear_mod} as well as Table~\ref{table:numerical_results4} indicate that the linear algorithms outmatch the non-linear approaches. This may be due to the fact that modulation increases the linearity of the pyramid sensor, and therefore linear algorithms are better suited. Usually, non-linear algorithms for solving inverse problems are more error-prone than linear methods. Due to, e.g., numerical errors, the less non-linear an underlying relation is, the better linear algorithms perform. Please note that in the linear algorithms, the CuReD \cite{Ros11,ZNRR11}, a reconstructor developed for Shack-Hartmann sensors, was applied for the first time steps of the AO loop in order to correct mainly for low frequencies. 

We infer that an application of non-linear reconstruction methods can notably improve the image quality when using a non-modulated pyramid sensor. For the modulated sensor, we recommend employing linear reconstruction algorithms at least as long as it is guaranteed that the residual phases being sensed by the wavefront sensor are small, e.g., in closed loop AO without significant non-common path errors of the system.

\begin{figure}[htbp]
\centering
\includegraphics[scale=0.65]{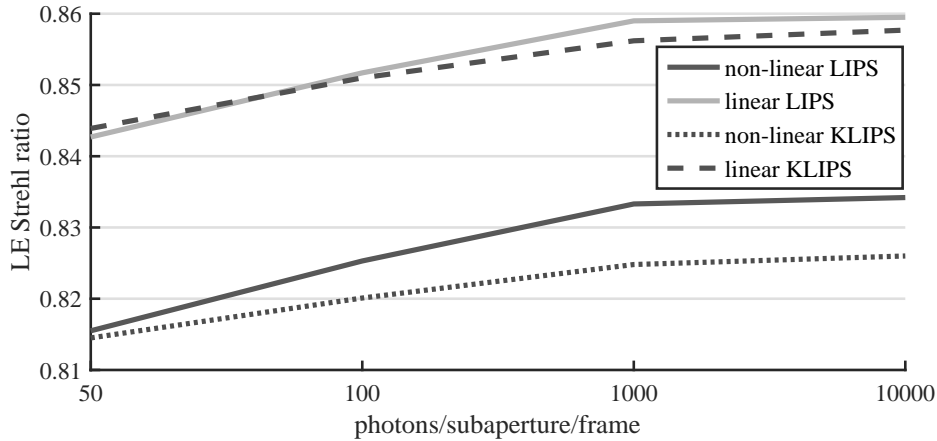}
\caption{Comparison of non-linear and linear methods for the modulated sensor: The linear algorithms clearly outmatch the non-linear processes. In the majority of cases the Landweber algorithms (solid lines) outperform the Landweber-Kaczmarz approaches (dashed lines).}
\label{fig:comparison_linear_mod}
\end{figure}

Concerning the computational complexity, the numerical effort of the proposed algorithm is comparable to the complexity of their linear versions being roughly $ \mathcal{O}\left(N_a^{3/2}\right)$ as derived in \cite{HuSha18_2}. Nevertheless, the amount of possible precomputations is larger in case of the linear versions of the algorithms which leads to a preference for linear methods with respect to the computational costs of online calculations.

\section{Conclusion and Outlook}

We have established two methods, namely Landweber and Landweber-Kaczmarz iteration for pyramid sensors, for accurate and stable non-linear wavefront reconstruction. The theoretical background was accompanied by a first numerical evaluation of the reconstruction quality. Especially for the non-modulated sensor, the two algorithms provided outstanding performance, with the Landweber iteration being outmatched by its Kaczmarz version.
 Although, the Landweber method is known to converge slowly, we experienced accelerated convergence for AO closed loop simulations using only a small amount of Landweber or Landweber-Kaczmarz iterations per time step. For the sensor having no modulation applied, the warm restart technique additionally sped up the convergence. The low number of necessary iterations positively impacts the computational load of the algorithms whose complexity is given by roughly $\mathcal{O}\left(N_a^{3/2}\right)$.
 
According to the results obtained when using linear and non-linear reconstruction methods in end-to-end simulations, we propose choosing the way of reconstruction, i.e., linear or non-linear, dependent on the modulation of the pyramid sensor. For the modulated sensor, we recorded higher reconstruction quality for linear reconstructors while for the non-modulated sensor we experienced better correction for the non-linear methods. However, note that this conclusion was drawn from a limited number of closed loop simulations without critical effects perturbing the linearity of the sensor such as NCPAs.

As already mentioned, the assumption of small residual wavefronts being measured by the wavefront sensor and on account of this the linearity of the pyramid sensor may be violated by non common path errors of the system. Thus, the ability of existing linear and non-linear reconstruction strategies to deliver high quality wavefront corrections even under the impact of large NCPAs is of great interest. We plan to come back to this topic including a detailed investigation of the reconstruction performance for the proposed non-linear algorithms LIPS and KLIPS in the presence of realistic NCPAs in a subsequent paper. Generally, the influence of the magnitude of the incoming phase distortions on the reconstruction quality needs to be analyzed for both linear and non-linear reconstruction methods.

Obstruction effects induced by wide telescope spiders (secondary mirror support structures) were omitted in the simulations. The non-linear LIPS and KLIPS both offer the possibility to be combined with direct segment piston reconstructors according to the so called Split Approach described in \cite{HuShaOb18,ObRafShaHu18_proc}. Nevertheless, the stability of the algorithms for segmented pupils needs to be examined in detail.

Proper choices for the basis representations~\eqref{eq:im.1}~-~\eqref{eq:im.2}, e.g., using wavelets, may significantly improve the reconstruction quality and will be analyzed in future work. However, the representation using the characteristic functions of the subapertures as basis functions allows for offline precomputations, an advantage which may not exist for different choices of basis representations.

The two wavefront reconstruction methods presented in this paper are based on the roof wavefront sensor forward model while data are obtained from pyramid sensors. The derivation of the Fr\'{e}chet derivatives and their adjoint operators for the full pyramid sensor model are planned. Since the latter more precisely describes the pyramid wavefront sensor, we may gain in reconstruction performance. We will come back to this topic including detailed comparisons of pyramid and roof sensor models as basis in the reconstruction approach in an upcoming paper. 

At present, there already exist several methods that allow for accurate wavefront reconstruction such as MVM based approaches or the P-CuReD algorithm \cite{Shat13} with LE Strehl ratios around $0.89$ as summarized in, e.g., \cite{HuShaOb18}. Nevertheless, we want to emphasize the importance of the development of new algorithms since the high reconstruction performance with the linear methods is obtained in undisturbed closed loop AO systems. As recently as extensive studies on the behavior of the linear and non-linear algorithms in the presence of realistic ELT effects such as NCPAs, telescope spiders or the low wind effect have been performed, further conclusions on preferences can be drawn. Due to the non-linearity of the pyramid sensor, degradation of the image quality may appear for reconstructors which are based on the linearity assumption. Furthermore, the LIPS and KLIPS can easier be adapted to more precise pyramid sensor models or segmented telescope pupils as, for instance, the P-CuReD algorithm. The non-linear algorithms are in an early stage of development and improvements for future investigations are expected. 

\section*{Funding}
Austrian Federal Ministry of Science and Research (HRSM); Austrian Science Fund (Project F68-N36).

\section*{Acknowledgments}
The authors are grateful to the Austrian Adaptive Optics team for fruitful discussion and thank Kirk Soodhalter for his support. 

\section*{Appendix}
\begin{proof}[Proof of Proposition \ref{roof_frechet_adjoint}]
Since the proof for the operator in $y$-direction is analogous to that of the one in $x$-direction, we only perform the calculation for $\left(\left(\boldsymbol R_x^{\{n,c\}}\right)'(\Phi)\right)^*$. For the evaluation of the adjoints, we divide the Fr\'{e}chet derivatives into two parts 
$$\left(\boldsymbol R_x^{\{n,c\}}\right)'(\Phi) = \left(\boldsymbol T_1^{\{n,c\}}(\Phi)\right) - \left(\boldsymbol T_2^{\{n,c\}}(\Phi)\right)$$ with
{\begin{align*}
\left(\boldsymbol T_1^{\{n,c\}}(\Phi)\right)\psi\left(x,y\right)
&:=\mathcal{X}_{\Omega}\left(x,y\right)\left[\dfrac{1}{\pi} \int_{\Omega_y}{\dfrac{\cos{\left[\Phi(x',y)-\Phi(x,y)\right]}\times k^{\{n,c\}}\left(x'-x\right)\psi(x',y)}{x'-x}\ \mathrm{d}x'}\right], \\
\left(\boldsymbol T_2^{\{n,c\}}(\Phi)\right) \psi\left(x,y\right)
&:=\mathcal{X}_{\Omega}\left(x,y\right)\left[\dfrac{1}{\pi}  \int_{\Omega_y}{\dfrac{\cos{\left[\Phi(x',y)-\Phi(x,y)\right]}\times k^{\{n,c\}}\left(x'-x\right)\psi(x,y)}{x'-x}\ \mathrm{d}x'}\right].
\end{align*}
}For any $\psi, \varphi \in \mathcal{L}_2\left(\mathbb{R}^2\right)$ with support on the telescope pupil $\Omega$ we consider
{\small\begin{align*}
\langle \left(\boldsymbol T_1^{\{n,c\}}(\Phi)\right)\psi,\varphi\rangle_{ \mathcal{L}_2\left(\mathbb{R}^2\right)} &= 
\langle \left(\boldsymbol T_1^{\{n,c\}}(\Phi)\right)\psi,\varphi\rangle_{ \mathcal{L}_2\left(\Omega\right)}  \\ 
&=
 \int_{\Omega_y} \int_{\Omega_x}\mathcal{X}_{\Omega}\left(x,y\right)\left[\dfrac{1}{\pi}  \int_{\Omega_y}{\dfrac{\cos{\left[\Phi(x',y)-\Phi(x,y)\right]}\times k^{\{n,c\}}\left(x'-x\right)\psi(x',y)}{x'-x}\ \mathrm{d}x'}\right] \\
 &\times \varphi\left(x,y\right) \ \mathrm{d}y \ \mathrm{d}x \\
&= 
\int_{\Omega_y} \int_{\Omega_x}\psi\left(x',y\right)\dfrac{1}{\pi}  \int_{\Omega_y}{\dfrac{\cos{\left[\Phi(x',y)-\Phi(x,y)\right]}\times k^{\{n,c\}}\left(x'-x\right)\varphi(x,y)}{x'-x}\ \mathrm{d}x} \ \mathrm{d}y \ \mathrm{d}x' \\
&= 
\int_{\Omega_y} \int_{\Omega_x}\psi\left(x,y\right)\dfrac{1}{\pi} \int_{\Omega_y}{\dfrac{\cos{\left[\Phi(x,y)-\Phi(x',y)\right]}\times k^{\{n,c\}}\left(x-x'\right)\varphi(x',y)}{x-x'}\ \mathrm{d}x'} \ \mathrm{d}y \ \mathrm{d}x \\
&=
 \langle \psi,\left(\boldsymbol T_1^{\{n,c\}}(\Phi)\right)^*\varphi\rangle_{ \mathcal{L}_2\left(\Omega\right)} \\
&=
 \langle \psi,\left(\boldsymbol T_1^{\{n,c\}}(\Phi)\right)^*\varphi\rangle_{ \mathcal{L}_2\left(\mathbb{R}^2\right)} 
\end{align*}
}with { $$\left(\left(\boldsymbol T_1^{\{n,c\}}(\Phi)\right)^*\varphi\right)(x,y) =- \mathcal{X}_{\Omega}\left(x,y\right)\dfrac{1}{\pi} \int_{\Omega_y}{\dfrac{\cos{\left[\Phi(x',y)-\Phi(x,y)\right]}\times k^{\{n,c\}}\left(x'-x\right)\varphi(x',y)}{x'-x}\ \mathrm{d}x'}$$ }using the fact that $k^{\{n,c\}}$ and cosine are even functions.
The adjoints of the second part are derived by
{\small\begin{align*}
\langle \left(\boldsymbol T_2^{\{n,c\}}(\Phi)\right)\psi,\varphi\rangle_{ \mathcal{L}_2\left(\mathbb{R}^2\right)} 
&=\langle \left(\boldsymbol T_2^{\{n,c\}}(\Phi)\right)\psi,\varphi\rangle_{ \mathcal{L}_2\left(\Omega\right)} \\
&=
 \int_{\Omega_y} \int_{\Omega_x}\mathcal{X}_{\Omega}\left(x,y\right)\left[\dfrac{1}{\pi}  \int_{\Omega_y}{\dfrac{\cos{\left[\Phi(x',y)-\Phi(x,y)\right]}\times k^{\{n,c\}}\left(x'-x\right)\psi(x,y)}{x'-x}\ \mathrm{d}x'}\right] \\
 &\times\varphi\left(x,y\right) \ \mathrm{d}y \ \mathrm{d}x \\
&= 
\int_{\Omega_y} \int_{\Omega_x}\psi\left(x,y\right) \dfrac{1}{\pi} \int_{\Omega_y}{\dfrac{\cos{\left[\Phi(x',y)-\Phi(x,y)\right]}\times k^{\{n,c\}}\left(x'-x\right)\varphi(x,y)}{x'-x}\ \mathrm{d}x'} \ \mathrm{d}y \ \mathrm{d}x \\
&=
 \langle \psi,\left(\boldsymbol T_2^{\{n,c\}}(\Phi)\right)^*\varphi\rangle_{ \mathcal{L}_2\left(\Omega\right)} \\
&=  \langle \psi,\left(\boldsymbol T_2^{\{n,c\}}(\Phi)\right)^*\varphi\rangle_{ \mathcal{L}_2\left(\mathbb{R}^2\right)}, 
\end{align*}
}which results in {$$\left(\left(\boldsymbol T_2^{\{n,c\}}(\Phi)\right)^*\varphi\right)(x,y) =\mathcal{X}_{\Omega}\left(x,y\right) \dfrac{1}{\pi} \int_{\Omega_y}{\dfrac{\cos{\left[\Phi(x',y)-\Phi(x,y)\right]}\times k^{\{n,c\}}\left(x'-x\right)\varphi(x,y)}{x'-x}\ \mathrm{d}x'}.$$
}Hence, the adjoints of the roof sensor Fr\'{e}chet derivatives are given by
{\small\begin{align*}
\left(\left(\boldsymbol R_x^{\{n,c\}}\right)'(\Phi)\right)^*\psi\left(x,y\right) &= \left(\left(\boldsymbol T_1^{\{n,c\}}(\Phi)\right)^*\psi\right)\left(x,y\right) -\left( \left(\boldsymbol T_2^{\{n,c\}}(\Phi)\right)^*\psi\right)\left(x,y\right) \\
&=-\mathcal{X}_{\Omega}\left(x,y\right)\dfrac{1}{\pi} \int_{\Omega_y}{\dfrac{\cos{\left[\Phi(x',y)-\Phi(x,y)\right]}\times k^{\{n,c\}}\left(x'-x\right)\left[\psi(x',y)+\psi(x,y)\right]}{x'-x}\ \mathrm{d}x'}.
\end{align*} }
\end{proof}
\bigskip \par

\begin{proof}[Proof of Proposition \ref{prop_discretization}]
Utilizing the representation~\eqref{eq:im.1} we obtain
{\begin{align*}
\left( \boldsymbol R_x^{\{n,c\}} \Phi\right) (x,y) 
&=\mathcal{X}_{\Omega}(x,y)   \dfrac { 1 } { \pi }\int_{\Omega_y} \dfrac{ \sin [ \Phi (x',y) - \Phi (x,y) ] \times k^{\{n,c\}} (x'-x) } {x'-x} \ \mathrm{d}x' \\
&=\mathcal{X}_{\Omega}(x,y)   \dfrac { 1 } { \pi } \int_{\Omega_y} \dfrac{ \sin [ \sum\limits_{i,j=1}^N\Phi_{ij}\mathcal{X}_{\Omega_{ij}}(x',y) - \sum\limits_{i,j=1}^N\Phi_{ij}\mathcal{X}_{\Omega_{ij}}(x,y) ] \times k^{\{n,c\}} (x'-x) } {x'-x} \ \mathrm{d}x' \\
&=\mathcal{X}_{\Omega}(x,y)   \dfrac { 1 } { \pi }  \int_{\Omega_y} \dfrac{ \sin [ \sum_{i,j=1}^N\Phi_{ij}\mathcal{X}_{\Omega_{ij}}(x',y) -\Phi_{ij}\mathcal{X}_{\Omega_{ij}}(x,y) ] \times k^{\{n,c\}} (x'-x) } {x'-x} \ \mathrm{d}x'.
\end{align*}
}The application of the non-linear roof sensor evaluated at a point $\left(x_m,y_n\right) \in \Omega_{mn} \subset \Omega_y\times\Omega_x$ is represented by
{\small\begin{align*}
\left( \boldsymbol R_x^{\{n,c\}} \Phi\right)\left(x_m,y_n\right) 
&=\mathcal{X}_{\Omega}(x_m,y_n)   \dfrac { 1 } { \pi }  \int_{\Omega_n} \dfrac{ \sin [ \sum\limits_{i,j=1}^N\Phi_{ij}\mathcal{X}_{\Omega_{ij}}(x',y_n) -\Phi_{ij}\mathcal{X}_{\Omega_{ij}}(x_m,y_n) ] \times k^{\{n,c\}} (x'-x_m) } {x'-x_m} \ \mathrm{d}x' \\
&=  \dfrac { 1 } { \pi }  \int_{\Omega_n} \dfrac{ \sin [ \sum_{i=1}^N\Phi_{in}\mathcal{X}_{\Omega_{in}}(x') -\Phi_{mn} ] \times k^{\{n,c\}} (x'-x_m) } {x'-x_m} \ \mathrm{d}x' 
\end{align*}
}for $m,n \in \mathbb{N}, 1 \le m,n\le N$. 
As the subapertures are disjoint, these considerations result in 
 \begingroup
 \allowdisplaybreaks
\begin{align*}
\left( \boldsymbol R_x^{\{n,c\}} \Phi\right)_{m,n} 
&=  \dfrac { 1 } { \pi }\int_{\Omega_n} \dfrac{ \sin [ \sum_{i=1}^N\Phi_{in}\mathcal{X}_{\Omega_{in}}(x') -\Phi_{mn} ] \times k^{\{n,c\}} (x'-x_m) } {x'-x_m} \ \mathrm{d}x'  \\
&=  \dfrac { 1 } { \pi } \int_{\Omega_{1n}} \dfrac{ \sin [ \sum_{i=1}^N\Phi_{in}\mathcal{X}_{\Omega_{in}}(x') -\Phi_{mn} ] \times k^{\{n,c\}} (x'-x_m) } {x'-x_m} \ \mathrm{d}x' \\
&+  \dfrac { 1 } { \pi }\int_{\Omega_{2n}} \dfrac{ \sin [ \sum_{i=1}^N\Phi_{in}\mathcal{X}_{\Omega_{in}}(x') -\Phi_{mn} ] \times k^{\{n,c\}} (x'-x_m) } {x'-x_m} \ \mathrm{d}x'  + \dots \\
&+  \dfrac { 1 } { \pi }\int_{\Omega_{Nn}} \dfrac{ \sin [ \sum_{i=1}^N\Phi_{in}\mathcal{X}_{\Omega_{in}}(x') -\Phi_{mn} ] \times k^{\{n,c\}} (x'-x_m) } {x'-x_m} \ \mathrm{d}x'  \\
&= \dfrac { 1 } { \pi }\int_{\Omega_{1n}} \dfrac{ \sin [ \Phi_{1n} -\Phi_{mn} ] \times k^{\{n,c\}} (x'-x_m) } {x'-x_m} \ \mathrm{d}x' \\
&+  \dfrac { 1 } { \pi }\int_{\Omega_{2n}} \dfrac{ \sin [\Phi_{2n} -\Phi_{mn} ] \times k^{\{n,c\}} (x'-x_m) } {x'-x_m} \ \mathrm{d}x'  + \dots \\
&+  \dfrac { 1 } { \pi }\int_{\Omega_{Nn}} \dfrac{ \sin [ \Phi_{Nn} -\Phi_{mn} ] \times k^{\{n,c\}} (x'-x_m) } {x'-x_m} \ \mathrm{d}x'  \\ 
&= \dfrac { 1 } { \pi }\ \sum_{\substack{i=1\\i \neq m}}^N \sin\left[\Phi_{in}-\Phi_{mn}\right]\int_{\Omega_{in}} \dfrac{  k^{\{n,c\}} (x'-x_m) } {x'-x_m} \ \mathrm{d}x' \\
&= \dfrac { 1 } { \pi }\ \sum_{\substack{i=1\\i \neq m}}^N \sin\left[\Phi_{in}-\Phi_{mn}\right] \alpha_{in}^{\{n,c\}}\left(x_m\right).
\end{align*}
\endgroup
 Analogously, for the adjoints of the Fr\'{e}chet derivatives using the representation~\eqref{eq:im.1} for both $\Phi$ and $\Psi$, we obtain
\begin{align*}
\left(\left(\boldsymbol R_x^{\{n,c\}}\right)'(\Phi)\right)^* \psi\left(x_m,y_n\right)
&=-\mathcal{X}_{\Omega}(x_m,y_n) \dfrac{1}{\pi}  \int_{\Omega_n}\dfrac{\cos{\left[\Phi(x',y_n)-\Phi(x_m,y_n)\right]}\times k^{\{n,c\}}\left(x'-x_m\right)}{x'-x_m}\\
&\times \left[\psi(x',y_n)+\psi(x_m,y_n)\right]\mathrm{d}x'\\
&=-\dfrac{1}{\pi} \int_{\Omega_n}{\dfrac{\cos{\left[\sum_{i,j=1}^N\Phi_{ij}\mathcal{X}_{\Omega_{ij}}(x',y_n)-\Phi_{ij}\mathcal{X}_{\Omega_{ij}}(x_m,y_n)\right]} }{x'-x_m}\ } \\
&\times k^{\{n,c\}}\left(x'-x_m\right)\left[\sum_{k,l=1}^N\psi_{kl}\mathcal{X}_{\Omega_{kl}}(x',y_n)+\psi_{kl}\mathcal{X}_{\Omega_{kl}}(x_m,y_n)\right] \ \mathrm{d}x' \\
&=-\dfrac{1}{\pi}  \int_{\Omega_n}{\dfrac{\cos{\left[\sum_{i=1}^N\Phi_{in}\mathcal{X}_{\Omega_{in}}(x')-\Phi_{mn}\right]}\times k^{\{n,c\}}\left(x'-x_m\right)}{x'-x_m}} \\
&\times \left[\sum_{k=1}^N\psi_{kn}\mathcal{X}_{\Omega_{kn}}(x')+\psi_{mn}\right] \ \mathrm{d}x'.
\end{align*}
For disjoint subapertures $\Omega_{ij}$, this results in
\begin{align*}
\left(\left(\left(\boldsymbol R_x^{\{n,c\}}\right)'(\Phi)\right)^* \psi\right)_{mn} &=
-\dfrac{1}{\pi}  \int_{\Omega_n} \dfrac{\cos{\left[\sum_{i=1}^N\Phi_{in}\mathcal{X}_{\Omega_{in}}(x')-\Phi_{mn}\right]}\times k^{\{n,c\}}\left(x'-x_m\right)}{x'-x_m}\\
&\times \left[\sum_{k=1}^N\psi_{kn}\mathcal{X}_{\Omega_{kn}}(x')+\psi_{mn}\right] \mathrm{d}x'  \\
&=-\dfrac{1}{\pi}\sum_{i=1}^N \cos{\left[\Phi_{in}-\Phi_{mn}\right]}\left[\psi_{in}+\psi_{mn}\right]  \int_{\Omega_{in}} \dfrac{ k^{\{n,c\}}\left(x'-x_m\right)}{x'-x_m}\ \mathrm{d}x'  \\
&=-\dfrac{1}{\pi}\sum_{i=1}^N \cos{\left[\Phi_{in}-\Phi_{mn}\right]}\left[\psi_{in}+\psi_{mn}\right] \alpha_{in}^{\{n,c\}}\left(x_m\right)
\end{align*}
with functions $\alpha_{in}^{\{n,c\}}$ defined in~\eqref{eq:im.4}. 
\end{proof}

\bibliographystyle{plain}
\bibliography{mainArXiV}

\end{document}